\documentclass[dvips,aap]{imsart}

\usepackage{latexsym,amsfonts,amsmath,amssymb,mathrsfs}
\usepackage{url}
\usepackage{graphicx}
\usepackage[usenames]{color}
\usepackage{euscript}
\usepackage{verbatim}
\usepackage{hyperref}
\usepackage{enumerate}

\startlocaldefs

\newtheorem{theorem}{Theorem}
\newtheorem{lemma}{Lemma}

\newtheorem{remark}{Remark}
\newtheorem{proposition}{Proposition}
\newtheorem{definition}{Definition}
\newtheorem{corollary}{Corollary}

\newtheorem{example}{Example}

\definecolor{darkgreen}{RGB}{0,100,0}
\definecolor{darkmagenta}{RGB}{139,0,139}
\definecolor{darkblue}{RGB}{40,0,200}

\newcommand{\jd}[1]{{\color{darkblue}{#1}}}

\newcommand{\rd}{\mathrm{d}}

\newcommand{\qed} {\hfill \Box \vspace{0.5cm}}

\newcommand{\EE}{\mathbb{E}}

\newcommand{\NN}{\mathbb{N}}

\newenvironment{proof}{\begin{trivlist}
    \item[\hskip\labelsep{\it Proof.}]}{$\hfill\Box$\end{trivlist}}

\newcommand{\abs}[1]{\left\vert #1 \right\vert}
\newcommand{\norm}[1]{\left\Vert #1 \right\Vert}

\newcommand{\R}{\mathbb{R}}
\renewcommand{\P}{\mathbb{P}}
\newcommand{\loc}{\text{\rm loc}}
\renewcommand{\a}{\alpha}

\endlocaldefs

\allowdisplaybreaks

\begin{document}

\begin{frontmatter}

\title{Discrepancy bounds for uniformly ergodic Markov chain quasi-Monte Carlo}
\runtitle{Uniformly ergodic Markov chain quasi-Monte Carlo}


\author{\fnms{Josef} \snm{Dick}\corref{Corresponding author}\ead[label=e1]{josef.dick@unsw.edu.au}}
\address{School of Mathematics and Statistics \\ The University of New South Wales \\ Sydney, NSW 2052, Australia \\ \printead{e1}}
\and
\author{\fnms{Daniel} \snm{Rudolf}\ead[label=e2]{daniel.rudolf@uni-jena.de}}
\address{
Inst. f\"ur Math.\\
Universit\"at Jena\\
Ernst-Abbe-Platz 2\\
07743 Jena,
Germany\\ \printead{e2}}
\and
\author{\fnms{Houying} \snm{Zhu}\ead[label=e3]{houying.zhu@unsw.edu.au}}
\address{School of Mathematics and Statistics \\ The University of New South Wales \\ Sydney, NSW 2052, Australia \\  \printead{e3}}

\runauthor{Dick, Rudolf, Zhu}

\begin{abstract}
Markov chains can be used to generate samples whose distribution approximates a given target distribution. The quality of the samples of such Markov chains can be measured by the discrepancy between the empirical distribution of the samples and the target distribution. We prove upper bounds on this discrepancy under the assumption that the Markov chain is uniformly ergodic and the driver sequence is deterministic rather than independent $U(0,1)$ random variables. In particular, we show the existence of driver sequences for which the discrepancy
of the Markov chain from the target distribution with respect to certain test sets converges with (almost) the usual Monte Carlo rate of $n^{-1/2}$.
\end{abstract}

\begin{keyword}[class=AMS]
\kwd[Primary ]{60J22, 65C40, 62F15}
\kwd[; secondary ]{65C05, 60J05}
\end{keyword}

\begin{keyword}
\kwd{Markov chain Monte Carlo}
\kwd{uniformly ergodic Markov chain}
\kwd{discrepancy theory}
\kwd{probabilistic method}
\end{keyword}

\end{frontmatter}

\section{Introduction}

Markov chain Monte Carlo (MCMC) algorithms are used for the approximation of
an  expected value with respect to the stationary probability measure $\pi$ of the chain.
This is done by simulating a Markov chain $(X_i)_{i\ge 1}$
and using the sample average $ \frac{1}{n} \sum_{i=1}^{n} f(X_i)$ to estimate the mean
$\mathbb{E}_{\pi}(f) := \int_{G} f(x) \pi(\rd x)$,
where $G$ is the state space and $f$ is a real-valued function defined on $G$.
This method is a staple tool in the physical sciences and Bayesian statistics.

A single transition from $X_{i-1}$ to $X_i$  of a Markov chain is generated by using
the current state $X_{i-1}$ and a random source $U_i$, usually taken from an i.i.d.
$\mathcal{U}(0,1)$ sequence $(U_i)_{i\ge 1}$ of random numbers.
In contrast, the Markov chain quasi-Monte Carlo idea is as follows:
Substitute the sequence of random numbers by a deterministically constructed finite sequence of
numbers $(u_i)_{1 \le i \le n}$ in $[0,1]^s$ for all $n \in \mathbb{N}$.
Numerical experiments suggest that for judiciously chosen deterministic pseudo-random numbers
$(u_i)_{1 \le i \le n}$ this can lead to significant improvements. Owen and Tribble \cite{OwTr05} and Tribble~\cite{Tr07} report an improvement by a factor of up to $10^3$ and a faster convergence rate for a Gibbs sampling problem. There were also previous attempts which provided evidence that the approach leads to comparable results \cite{LeSi06,Liao,So74}.
Another line of research, dealing with the so-called array-RQMC method, also combines MCMC with quasi-Monte Carlo \cite{LeLeTu08}. For a thorough literature review we refer to \cite[Subsection 1.1 (Literature review)]{ChDiOw11}.

Recently in the work of Chen, Dick and Owen \cite{ChDiOw11} and Chen~\cite{Ch11},
the first theoretical justification of the Markov chain quasi-Monte Carlo
approach on continuous state spaces was provided.
Therein a consistency result is proven if the random sequence $(U_i)_{i\ge 1}$ is substituted
by a deterministic `completely uniformly distributed' (CUD) sequence $(u_i)_{i\ge 1}$, the Markov chain satisfies a contraction assumption and the integrand $f$ is continuous. For a precise definition of CUD sequences we refer to \cite{ChDiOw11,ChMaNiOw12} and for
the construction of weakly CUD sequences we refer to \cite{TrOw08}.
The consistency result of Markov chain quasi-Monte Carlo
corresponds to an ergodic theorem for Markov chain Monte Carlo and can be shown to be equivalent to the statement that the discrepancy between the empirical distribution and target distribution converges to $0$ (this follows directly from \cite[Theorem~1]{ChDiOw11}). However, from the result in \cite{ChDiOw11} it is not clear
how fast the sample average converges to the desired expectation.
The goal of this paper is to investigate the convergence behavior
of such Markov chain quasi-Monte Carlo algorithms.
We describe the setting and main results in the following.

Throughout the paper we deal with uniformly ergodic Markov chains on a state
space $G \subseteq \mathbb{R}^d$ and a probability space $(G, \mathcal{B}(G), \pi)$,
where $\mathcal{B}(G)$ is the Borel $\sigma$-algebra defined on $G$ and $\pi$ is the
stationary distribution of the Markov chain, for details see for example \cite{MeTw09, RoCa04, RoRo04}. We assume that the Markov chain can be generated by an update function $\varphi:G \times [0,1]^s \to G$, that is, $X_i = \varphi(X_{i-1}; U_i)$ for all $i \ge 1$.
We fix a starting point $x_0 = x$ and replace the random numbers $(U_i)_{i \ge 1}$
by a deterministic sequence $(u_i)_{i\ge 1}$ to generate the deterministic points
$x_i = \varphi(x_{i-1}; u_i)$ for $i \ge 1$. The convergence behavior of the Markov chain
is measured using a generalized Kolmogorov-Smirnov test between the stationary distribution
$\pi$ and the empirical distribution $\widehat{\pi}_n(A) := \frac{1}{n} \sum_{i=1}^n 1_{x_i \in A}$,  where $1_{x_i \in A}$ is the indicator function
of the set $A \in \mathcal{B}(G)$. The discrepancy is defined by taking the supremum of $|\pi(A) - \widehat{\pi}_n(A)|$ over all sets in $\mathscr{A} \subseteq \mathcal{B}(G)$
(since the empirical distribution is based on a finite number of points in $G$
we generally have $\mathscr{A} \neq \mathcal{B}(G)$, see below for a more detailed description).
Under these assumptions we prove that, for each $n \in \mathbb{N}$, there exists a finite sequence of numbers $(u_i)_{1 \le i \le n}$ such that this discrepancy converges with order $\mathcal{O}(n^{-1/2} (\log n)^{1/2})$ as $n$ tends to infinity.
This is roughly the convergence rate which one would expect from MCMC algorithms based on random inputs.

A drawback of our results is that we are currently not able to give explicit constructions of sequences $(u_i)_{1 \le i \le n}$ for which our discrepancy bounds hold. This is because our proofs make essential use of probabilistic arguments. Namely, we use a
Hoeffding inequality by Glynn and Ormoneit~\cite{GlOr02}  and some results by  Talagrand~\cite{Ta94} on empirical processes and a result by Haussler~\cite{Ha95}. 
Roughly speaking, we use the Hoeffding inequality to show that the probability of all $(X_i)_{1\leq i\leq n}$ with small discrepancy is bigger than $0$,  which implies the existence of a Markov chain with small discrepancy.
We do, however, give a criterion (which we call `pull-back discrepancy') which the numbers $(u_i)_{1 \le i \le n}$ need to satisfy
such that the point set $(x_i)_{1 \le i \le n}$ has small discrepancy.
This is done by showing that the discrepancy of $(x_i)_{1\le i \le n}$ is
close to the pull-back discrepancy of the driver sequence $(u_i)_{1\le i \le n}$.
This should eventually lead to explicit constructions of suitable driver sequences.
As a corollary to the relation between the discrepancy of the Markov chain and the
pull-back discrepancy of the driver sequence, we obtain a Koksma-Hlawka inequality
for Markov chains in terms of the discrepancy of the driver sequence. We point out that the 
pull-back discrepancy generally differs from the CUD property studied in \cite{Ch11} and \cite{ChDiOw11}. Convergence rates beyond the usual Monte Carlo rate of $n^{-1/2}$ have previously  been shown for Array-RQMC \cite{LeLeTu08} and  in \cite[Chapter~6]{Ch11}. In both of these instances,  a direct simulation is (at least in principle) possible.

Our results on the discrepancy of the points $(x_i)_{1 \le i \le n}$ can also be understood as an extension of results on point distributions in the unit cube $[0,1]^s$, see \cite{HeNoWaWo01}, to uniformly ergodic Markov chains.

We give a brief outline of our work.
In the next section we provide background information on uniformly ergodic Markov chains,
give a relation between the transition kernel of a Markov chain
and their update function and state some examples which satisfy
the convergence properties.
We also give some background on discrepancy and describe our results in more detail. In Section~\ref{sec_discr} we provide the notion of discrepancy with respect to the driver sequence and we prove the close relation between the two types of discrepancy for uniformly ergodic Markov chains from which we deduce a Koksma-Hlawka type inequality. In Section~\ref{sec_exist} we prove the main results.
The appendix contains sections on $\delta$-covers, the integration error and some technical proofs.

\section{Background and notation}

In this section we provide the necessary background on discrepancy
and uniformly ergodic Markov chains.

\subsection{Discrepancy}

The convergence behavior of the Markov chain is analyzed with respect to a distance measure between the empirical distribution of the Markov chain and its stationary distribution $\pi$. It can be viewed as an extension of the Kolmogorov-Smirnov test and is a well established concept in numerical analysis and number theory~\cite{DiPi10}. We analyze the empirical distribution of the first $n$ points of the Markov chain $X_1,\ldots, X_n$ by assigning each point the same weight and defining the empirical measure of a set $A \in \mathcal{B}(G)$ by
\begin{equation*}
\widehat{\pi}_n(A) = \frac{1}{n} \sum_{i=1}^n 1_{X_i \in A},
\end{equation*}
where the indicator function is given by
$$1_{X_i \in A} = \left\{\begin{array}{rl} 1 & \mbox{if } X_i \in A, \\ 0 & \mbox{otherwise}. \end{array} \right.$$ The local discrepancy between the empirical distribution and the stationary distribution is then
\begin{equation*}
\Delta_{n,A} = \widehat{\pi}_n(A) - \pi(A).
\end{equation*}
To obtain a measure for the discrepancy
we take the supremum of $|\Delta_{n,A}|$ over certain sets $A$.
Note that since the empirical measure uses only a finite number
of points the local discrepancy $\Delta_{n,A}$ does not converge to $0$
in general if we take the supremum over all sets in $\mathcal{B}(G)$.
Thus we restrict the supremum to a set of so-called test sets $\mathscr{A} \subseteq \mathcal{B}(G)$.
Now we define the discrepancy.
\begin{definition}[Discrepancy]
 The discrepancy of $P_n = \{X_1,\ldots, X_n\} \subseteq G$ is given by
\begin{equation*}
D^\ast_{\mathscr{A}, \pi}(P_n) = \sup_{A \in \mathscr{A}} |\Delta_{n,A}|.
\end{equation*}
\end{definition}
This is the measure which we use to analyze the convergence behavior of the Markov chain as $n$ goes to $\infty$.

In Appendix~\ref{sec_int_error} we provide a relationship between the discrepancy
$D^\ast_{\mathscr{A}, \pi}(P_n)$ and the integration error of functions
in a certain function space $H_1$, where the set of test sets is given by
$$\mathscr{A} = \{ (-\infty, x)_G : x \in \bar{\mathbb{R}}^d\},$$ with
$\bar{\mathbb{R}}^d= (\mathbb{R}\cup \{\infty,-\infty\})^d$ and
$(-\infty, x)_G := (-\infty, x) \cap G = \prod_{j=1}^d (-\infty, \xi_j)  \cap G$ for
$x = (\xi_1,\ldots, \xi_d)$.
In particular, if there is at least one $i$ with $\xi_i = -\infty$, then $(-\infty, x)_G = \emptyset$,
whereas if all $\xi_i = \infty$, then $(-\infty, x)_G = G \subseteq \mathbb{R}^d$. For functions $f \in H_1$ we have
\begin{equation*}
\left|\mathbb{E}_\pi(f) - \frac{1}{n} \sum_{i=1}^n f(X_i) \right| \le D^\ast_{\mathscr{A}, \pi}(P_n) \|f\|_{H_1}.
\end{equation*}
Inequalities of this form are called Koksma-Hlawka inequalities, see \cite[Chapter~2]{DiPi10} for more information. See Appendix~\ref{sec_int_error} for details on the definition of the space $H_1$ and the proof of the inequality.

\subsection{Markov chains}\label{sec_mc}

The main assumption on the Markov chain in \cite{Ch11}
and \cite{ChDiOw11} is the existence of a coupling region, or in a weakened version,
a contraction assumption on the update function. Roughly speaking, this means that if
one starts two Markov chains at different starting points but uses the same random numbers as updates,
then the points of the chain coincide or move closer to each other as the chain progresses.
In this paper, we replace this assumption by the assumption that the Markov chain is uniformly ergodic.
The concept of uniform ergodicity is much closer to the concept of discrepancy,
which allows us to obtain stronger results than previous attempts.
We introduce uniformly ergodic Markov chains in the following.

Let $G \subseteq \mathbb{R}^d$ and let $\mathcal{B}(G)$ denote the Borel $\sigma$-algebra of $G$.
In the following we provide the definition of a transition kernel.
\begin{definition}
 The function $K\colon G\times \mathcal{B}(G) \to [0,1]$ is called transition kernel if
 \begin{enumerate}[(i)]
  \item for each $x\in G$ the mapping $A\in\mathcal{B}(G) \mapsto K(x,A)$ is a probability measure on $(G,\mathcal{B}(G))$, and
  \item for each $A\in\mathcal{B}(G)$ the mapping $x\in G \mapsto K(x,A)$ is a $\mathcal{B}(G)$-measurable real-valued function.
 \end{enumerate}
\end{definition}

Let $K:G \times \mathcal{B}(G) \to [0,1]$ be a transition kernel.
We assume that $\pi$ is the unique stationary distribution of the transition kernel $K$, i.e.
\[
 \int_G K(x,A) \pi(\rd x) = \pi(A),\qquad \forall A\in \mathcal{B}(G).
\]
The transition kernel $K$ gives rise to a Markov chain $X_0, X_1, X_2, \ldots \in G$
in the following way. Let $X_0=x$ with $x\in G$ and $i\in\NN$.
Then, for a given $X_{i-1}$, we choose $X_i$ with distribution
$K(X_{i-1}, \cdot)$, that is, for all $A \in \mathcal{B}(G)$,
the probability that $X_i \in A$ is given by $K(X_{i-1}, A)$.

\begin{definition}[Total variation distance]
The total variation distance between
the transition kernel $K(x,\cdot)$ and the stationary distribution $\pi$ is defined by
\[
 \norm{K^j(x,\cdot)-\pi}_{\text{\rm tv}} = \sup_{A\in \mathcal{B}(G) } \abs{K^j(x,A)-\pi(A)}.
\]
\end{definition}

Note that with $K^0(x,A)=1_{x\in A}$ we have
\[
 K^j(x,A) = \int_G K(y,A)\,K^{j-1}(x,\rd y)
          =  \int_G K^{j-1}(y,A)\,K(x,\rd y).
\]

\begin{definition}[Uniform ergodicity]
Let $\a \in [0,1)$ and $M\in(0,\infty)$. The transition kernel $K$
is uniformly ergodic with $(\a,M)$ iff for any $x\in G$ and $j\in\NN$ we have
\[
 \norm{K^j(x,\cdot)-\pi}_{\text{\rm tv}} \leq \a^j M.
\]
A Markov chain with transition kernel $K$ is called uniformly ergodic if there exists an
$\a \in [0,1)$ and $M\in(0,\infty)$, such that the transition kernel is uniformly ergodic with $(\a,M)$.
\end{definition}
\begin{remark}
The uniform ergodicity is necessary to apply the Hoeffding inequality \cite{GlOr02}
and it is also used in the estimates of the discrepancy.

 For Markov chains
 that are
 geometrically ergodic or have a spectral gap,
 for definitions see \cite{RoRo97,RoRo04,Ru12},
 one must use other concentration inequalities.
  The papers \cite{Ad08,Mi12,Pa12}
 might be useful.
However, in those cases we are not aware of
  results which allow us to treat Markov chains
  which start deterministically, as we consider in this paper.
 \end{remark}

Let us state a result which provides an equivalent statement to uniform ergodicity.
Let $L_\infty$ be the set of all bounded functions $f\colon G\to \R$.
Then define the operator $P^j\colon L_\infty \to L_\infty$ by
\[
 P^j f(x) = \int_G f(y) \,K^j(x,\rd y),
\]
and the expectation with respect to $\pi$ is denoted by $\mathbb{E}_\pi(f)=\int_G f(y) \pi(\rd x)$.
The following result is well known, for a proof of this fact see for example \cite[Proposition~3.23. p.~48]{Ru12}.
\begin{proposition}  \label{prop: L_infty_erg}
   Let $\a \in[0,1)$ and $M\in(0,\infty)$.
   Then the following statements are equivalent:
    \begin{enumerate}[(i)]
  \item The transition kernel $K$ is uniformly ergodic with $(\a,M)$.
  \item The operator $P^j-\mathbb{E}_\pi$ satisfies
        \[
         \norm{P^j-\mathbb{E}_\pi}_{L_\infty \to L_\infty} \leq 2M\a^j , \quad j\in\NN.
        \]
  \end{enumerate}
\end{proposition}

In the following we introduce
update functions $\varphi$ for a given transition kernel and state some examples.
\begin{definition}[Update function]
Let $\varphi:G \times [0,1]^s \to G$
be a measurable
function and
\begin{align*}
B : G \times \mathcal{B}(G) & \to \mathcal{B}([0,1]^s), \\
B(x,A) & = \{u \in [0,1]^s: \varphi(x;u) \in A\},
\end{align*}
where $\mathcal{B}([0,1]^s)$ is the Borel $\sigma$-algebra of $[0,1]^s$. Let $\lambda_s$ denote the Lebesgue measure on $\mathbb{R}^s$.  Then the function $\varphi$ is an update function for the transition kernel $K$ iff
\begin{equation}\label{eq_update_prop}
K(x,A) = \P(\varphi(x;U)\in A) = \lambda_s(B(x,A)),
\end{equation}
where $\P$ is the probability measure for the uniform distribution in $[0,1]^s$.
\end{definition}

\begin{example}(Direct simulation) \label{ex: simple_rem}
Let us assume that we can sample with respect to $\pi$, i.e. $K(x,A) =\pi(A)$ for all $x \in G$.
For the moment let $G=[0,1]^s$ and let $\pi$ be the uniform distribution on $G$.
In this case we can choose the simple update function $\varphi(x; u) = u$,
since then
\begin{equation*}
\pi(A)
= \lambda_s(B(x,A)) \quad \mbox{for all } x \in G.
\end{equation*}
If $G$ is a general subset of $\R^d$ and $\pi$ is a general probability measure,
then we need a generator, see \cite{ChDiOw11}. A generator
is a special update function $\psi\colon [0,1]^s \to G$ such that
\[
  \pi(A)=\P(\psi(U)\in A), \quad \mbox{for all } A\in\mathcal{B}(G).
\]
Note that the transition kernel $K(x,A) =\pi(A)$ is uniformly ergodic with $(\a,M)$
for $\a=0$ and $M\in(0,\infty)$.
 \end{example}

\begin{example}(Hit-and-run algorithm) \label{ex: har}
Let $G\subset \R^d$ be a compact convex body and $\pi$ be the uniform distribution on $G$.
Let $\mathbb{S}^{d-1}=\{x\in \R^d\colon \norm{x}_2= \langle x , x \rangle^{1/2}=1\}$
be the $d-1$-dimensional sphere, where
$\langle x , y \rangle$ denotes the standard inner product in $\R^d$.
Let $\theta \in \mathbb{S}^{d-1}$ and let $L(x,\theta)$ be the chord in $G$ through $x$ and $x+\theta$,
i.e.
\[
 L(x,\theta) = \{ x+s\theta\in \R^d \mid s\in\R \} \cap G.
\]
We assume that we have an oracle which gives us $a(x,\theta),b(x,\theta) \in G$,  such that
\[
 [a(x,\theta),b(x,\theta)]= L(x,\theta) ,
\]
 where $[a(x,\theta),b(x,\theta)]=\{  \lambda a(x,\theta)+(1-\lambda) b(x,\theta)\colon \lambda \in[0,1]\}$.
 A transition of the hit and run algorithm works as follows. First, choose a random direction $\theta$. Then we sample the next state on $[a(x,\theta),b(x,\theta)]$ uniformly.
 Let $\psi\colon [0,1]^{d-1} \to \mathbb{S}^{d-1}$ be a generator for the uniform distribution on the sphere,
see for instance \cite{Fang_Wang}. Then we can choose for $x\in G$ and $u = (\upsilon_1, \upsilon_2, \ldots, \upsilon_{d}) \in[0,1]^{d}$
 the update function
 \[
  \varphi(x,u) = \upsilon_{d} \, a(x,\psi(\upsilon_1\dots,\upsilon_{d-1})) + (1-\upsilon_{d})\, b(x,\psi(\upsilon_1,\dots, \upsilon_{d-1})).
 \]
 In \cite{Sm84} it is shown that there exists an $\a\in[0,1)$ and an $M\in(0,\infty)$, such that
 the hit-and-run algorithm is uniformly ergodic with $(\a,M)$.
\end{example}

\begin{example}(Independence Metropolis sampler) \label{ex: ind_metro_sampler}
 Let $G=[0,1]^d$, assume that a parameter $\beta>0$
 and a function $H\colon [0,1]^d \to \R$ are given.
 Let $\pi_\beta$ be a probability measure on $G$, given by
 \[
  \pi_\beta (A)
  = \frac{1}{Z_\beta} \int_A \exp(-\beta H(x))\,\rd x, \qquad A\in \mathcal{B}(G),
 \]
 with normalizing constant $Z_\beta = \int_G \exp(-\beta H(y))\, \rd y$.
 Here $\pi_\beta$ might be interpreted as Boltzmann distribution with
 inverse temperature $\beta$ and Hamiltonian $H$.
 Now let $\bar u \in [0,1]^d$ and
  \[
   A(x;\bar u)= \min\left\{ 1 ,
   \exp(-\beta(H(\bar u)-H(x)))
   \right\}
  \]
  be the acceptance probability of the Metropolis transition.
  Then we can choose for $x\in [0,1]^d$
  and $u = (\upsilon_1, \upsilon_2, \ldots, \upsilon_{d+1}) \in[0,1]^{d+1}$
  the update function
  \[
   \varphi(x;u) = \begin{cases}
                    (\upsilon_1,\dots,\upsilon_d) & \upsilon_{d+1} \leq A(x;\upsilon_1,\dots, \upsilon_d), \\
                    x, &   \upsilon_{d+1} > A(x;\upsilon_1,\dots, \upsilon_d).
                  \end{cases}
  \]
In \cite[Theorem~2.1., p.~105]{MeTw96} it is proven that, if there is a number
$\gamma>0$ such that $\exp(\beta \inf_{x\in G} H(x)) \geq \gamma$, then
the independent Metropolis sampler is uniformly ergodic with $(1-\gamma,1)$.
A local proposal Metropolis algorithm can also be uniformly ergodic,
see for example \cite{MaNo07}.
 \end{example}

Note that the arguments of Example~\ref{ex: ind_metro_sampler} 
are also valid for a heat-bath sampler.
Let us briefly add some more examples.
 The slice sampler, for details with respect to the algorithm and update functions see \cite{Ne03},
 is under additional assumptions uniformly ergodic, see \cite{MiTi02}.
 Furthermore, the Gibbs sampler for sampling the uniform distribution is uniformly ergodic if the boundary of $G$
 is smooth enough, see \cite{RoRo98}.\\

Above we defined the set $B(x,A)$, which is for $x\in G$ and $A\in \mathcal{B}(G)$
the set of random numbers $u$ which takes $x$ into the set $A$ using the update function $\varphi$
with arguments $x$ and $u$. We now define sets of random numbers which take $x$ to $A$ in $i\in\NN$ steps. Let $\varphi_1(x;u) = \varphi(x;u)$ and for $i > 1$ let
\begin{align*}
\varphi_i & : G \times [0,1]^{is} \to G, \\
\varphi_i(x; u_1, u_2, \ldots, u_i) & = \varphi(\varphi_{i-1}(x; u_1, u_2,\ldots, u_{i-1}); u_i),
\end{align*}
that is, $\varphi_i(x; u_1, u_2,\ldots, u_i) \in G$ is the point obtained via $i$ updates using $u_1, u_2, \ldots, u_i \in [0,1]^s$ where the starting point is $x \in G$.

\begin{lemma}  \label{lem: update_fct}
 Let $i,j\in\NN$ and $i\geq j$. For any $u_1,\dots,u_i \in [0,1]^s$ and $x\in G$ we have
 \begin{equation} \label{eq: iterat_update_fct}
    \varphi_i(x;u_1,\dots,u_i) = \varphi_{i-j}(\varphi_j(x;u_1,\dots,u_j);u_{j+1},\dots,u_i).
  \end{equation}
\end{lemma}
\begin{proof}
 The assertion can be proven by induction over $i$.
\end{proof}

For $i \ge 1$ let
\begin{align*}
B_i & : G \times \mathcal{B}(G) \to \mathcal{B}([0,1]^{i s}), \\
B_i(x,A) & = \{(u_1, u_2, \ldots, u_i) \in [0,1]^{i s}: \varphi_i(x; u_1, u_2,\ldots, u_i) \in A\}.
\end{align*}
We therefore have $B_1(x,A) = B(x,A)$. Note that $B_i(x,A) \subseteq [0,1]^{is}$.
The next lemma is important to understand the relation
between the update function and the transition kernel.
\begin{lemma}  \label{lem: same_expect}
Let $\varphi$ be an update function for the transition kernel $K$. Let $n\in \NN$
and $F\colon G^n \to \R$. The expectation with respect to the joint distribution of
$X_1,\dots,X_n$ from the Markov chain starting at $x_0\in G$ is given by
\[
 \EE_{x_0,K}(F(X_1,\dots,X_n))
= \underbrace{\int_G \dots \int_G}_{n\text{-times}}
F(x_1,\dots,x_n)\, K(x_{n-1},\rd x_n) \dots K(x_0,\rd x_1).
\]
Then
\begin{equation} \label{eq: same_expect}
 \begin{split}
&   \EE_{x_0,K}(F(X_1,\dots,X_n)) \\
&  \qquad =\int_{[0,1]^{ns}} F(\varphi_1(x_0,u_1),\dots,\varphi_n(x_0,u_1,\dots,u_n)) \rd u_1\dots \rd u_n,
  \end{split}
\end{equation}
whenever one of the integrals exist.
\end{lemma}
Note that the right-hand-side of \eqref{eq: same_expect} is the expectation with respect
to the uniform distribution in $[0,1]^{ns}$.

\begin{proof}
First note that by the definition of
the update function we obtain for any $\pi$-integrable function $f\colon G \to \R$ that
\begin{equation}  \label{eq: trans_op_up_fct}
  \int_G f(y) K(x,\rd y) = \int_{[0,1]^s} f(\varphi(x,u)) \rd u.
\end{equation}
By the application of Lemma~\ref{lem: update_fct} and \eqref{eq: trans_op_up_fct} we obtain
\begin{align*}
   & \int_{[0,1]^{ns}} F(\varphi_1(x_0,u_1),\dots,\varphi_n(x_0,u_1,\dots,u_n)) \rd u_1\dots \rd u_n \\
 = & \int_{[0,1]^{(n-1)s}} \int_G F(x_1,\varphi_1(x_1,u_2),\dots,\varphi_{n-1}(x_1,u_2,\dots,u_n)) \times\\
& \qquad  \qquad \qquad \qquad \qquad \qquad \qquad \qquad \qquad
         K(x_0,\rd x_1) \rd u_2\dots \rd u_n.
\end{align*}
The iteration of this procedure leads to the assertion.
\end{proof}

\begin{corollary} \label{coro: stat_update}
Let $\varphi$ be an update function for the transition kernel $K$ and let $\pi$ be the stationary distribution of $K$. For any $i \in \mathbb{N}$
and $A \in \mathcal{B}(G)$ we have $K^i(x,A) = \lambda_{is}(B_i(x,A))$.
In particular
\begin{equation*}
\int_G \lambda_{is}(B_i(x, A))\, \pi(\rd x) = \pi(A).
\end{equation*}
\end{corollary}
\begin{proof}
 Set for $n\geq i$
 \[
  F(x_1,\dots,x_n) = 1_{x_i \in A}.
 \]
 Then by Lemma~\ref{lem: same_expect} we obtain $K^i(x,A) = \lambda_{is}(B_i(x,A))$
 and by the stationarity of $\pi$ the proof is complete.
\end{proof}

\section{On the discrepancies of the Markov chain and driver sequence}\label{sec_discr}

Recall that the star-discrepancy of a point set
$P_{n} = \{x_1, x_2, \ldots, x_n \} \subseteq G$ with respect to the distribution $\pi$
is given by
\begin{equation*}
D^\ast_{\mathscr{A}, \pi}(P_n)
= \sup_{A \in \mathscr{A}} \left|\frac{1}{n} \sum_{i=1}^n 1_{x_i \in A} - \pi(A)\right|.
\end{equation*}

Let us assume that $u_1,u_2,\ldots, u_n \in [0,1]^s$ is a finite deterministic sequence. We call this finite sequence driver sequence. Then let the set $P_n=\{x_1, x_2, \ldots, x_n \} \subseteq G$ be given by
\begin{equation} \label{eq: x_i_by_driver_seq}
  x_i =  \varphi(x_{i-1};u_i)  = \varphi_i(x_0;u_1,\dots,u_i), \quad i=1,\dots,n.
\end{equation}

We now define a discrepancy measure on the driver sequence. Below we show how this discrepancy is related to the discrepancy of the Markov chain.

\begin{definition}[Pull-back discrepancy]
Let $\mathcal{U}_n = (u_1, u_2,\ldots, u_n) \in [0,1]^{ns}$ and let $B_i$ be
defined as above. Define the local discrepancy function by
\[
\Delta^{\loc}_{n,A,\varphi}(x; u_1,\dots,u_n)
=\frac{1}{n}  \sum_{i=1}^n \left[1_{(u_1,\ldots, u_i) \in B_i(x,A)} - \lambda_{is}(B_i(x,A)) \right].
\]
Let $\mathscr{A} \subseteq \mathcal{B}(G)$ be a set of test sets. Then we define the discrepancy of the driver sequence by
\begin{equation*}
D^\ast_{\mathscr{A}, \varphi}(\mathcal{U}_n) = \sup_{A \in \mathscr{A}} \left|\Delta^{\loc}_{n,A,\varphi}(x; u_1,\dots,u_n) \right|.
\end{equation*}
We call $D^\ast_{\mathscr{A}, \varphi}(\mathcal{U}_n)$ the pull-back discrepancy.
\end{definition}

The discrepancy of the driver sequence $D^\ast_{\mathscr{A}, \varphi}(\mathcal{U}_n)$
is a `pull-back discrepancy' since the test sets $B_i(x,A)$ are derived
from the test sets $A \in \mathscr{A}$ from the discrepancy of the Markov chain $D^\ast_{\mathscr{A}, \pi}(P_n)$ via inverting the update function.

The following theorem provides an estimate of the star-discrepancy of $P_n$ with respect to
properties of the driver sequence and the transition kernel.

\begin{theorem} \label{thm: est_discr}
Let $K$ be a transition kernel defined on $G \subseteq \mathbb{R}^d$
with stationary distribution $\pi$.
Let $\varphi$ be an update function for $K$.
Let $x_0=x$ and let $u_1, u_2, \ldots, u_n \in [0,1]^s$ be the driver sequence, such
that $P_n$ is given by \eqref{eq: x_i_by_driver_seq}.
Let $\mathscr{A} \subseteq \mathcal{B}(G)$ be a set of test sets. Then
\begin{align*}
\abs{D^\ast_{\mathscr{A}, \pi}(P_n) - D^\ast_{\mathscr{A},\varphi}(\mathcal{U}_n) }
&  \le \sup_{A \in \mathscr{A}} \abs{  \frac{1}{n} \sum_{i=1}^n K^i(x,A) - \pi(A) }.
\end{align*}

\end{theorem}
\begin{proof}
For any $A \in \mathscr{A}$ we have
\begin{align*}
& \quad\, \left| \frac{1}{n} \sum_{i=1}^n 1_{x_i \in A} - \pi(A) \right| \\
&= \left| \frac{1}{n} \sum_{i=1}^n \left[ 1_{(u_1,\ldots, u_i) \in B_i(x,A)} - K^i(x,A) + K^i(x,A)- \pi(A) \right] \right| \\
 & \leq \abs{ \frac{1}{n} \sum_{i=1}^n \left[1_{(u_1,\ldots, u_i) \in B_i(x,A)} - \lambda_{is}(B_i(x,A))\right]}
+ \abs{ \frac{1}{n} \sum_{i=1}^n K^i(x,A) - \pi(A) }.
\end{align*}
 Note that we used $\lambda_{is}(B_i(x,A))=K^i(x,A)$
 which follows from Corollary~\ref{coro: stat_update}. Hence
\[
 D^\ast_{\mathscr{A}, \pi}(P_n)
  \le D^\ast_{\mathscr{A}, \varphi}(\mathcal{U}_n) + \sup_{A \in \mathscr{A}} \abs{  \frac{1}{n} \sum_{i=1}^n K^i(x,A) - \pi(A) }.
 \]
 The inequality
 \[
D^\ast_{\mathscr{A}, \varphi}(\mathcal{U}_n)
  \le D^\ast_{\mathscr{A}, \pi}(P_n)  +
\sup_{A \in \mathscr{A}} \abs{  \frac{1}{n} \sum_{i=1}^n K^i(x,A) - \pi(A) }
 \]
 follows by the same arguments.
\end{proof}

\begin{corollary}  \label{coro: D_U_almost_D_P}
 Let us assume that the conditions of Theorem~\ref{thm: est_discr} are satisfied.
 Further let $\a\in[0,1)$ and $M\in(0,\infty)$ and assume that the transition kernel is uniformly
 ergodic with $(\a,M)$. Then
 \begin{align*}
\left|D^\ast_{\mathscr{A}, \pi}(P_n) - D^\ast_{\mathscr{A}, \varphi}(\mathcal{U}_n) \right| \le &
\frac{\a M }{n(1-\a)}.
\end{align*}
\end{corollary}
\begin{proof}
  By the uniform ergodicity with $(\a,M)$ we obtain
  \begin{align*}
   \abs{  \frac{1}{n} \sum_{i=1}^n K^i(x,A) - \pi(A) }
& \leq \frac{1}{n} \sum_{i=1}^n \abs{K^i(x,A) - \pi(A) } \\
& \leq \frac{1}{n} \sum_{i=1}^\infty \a^i M = \frac{\a M }{n(1-\a)}.
  \end{align*}
  Then by Theorem~\ref{thm: est_discr} the assertion is proven.
\end{proof}

\begin{remark} \label{rem: simple_unif_erg}
 In the setting of Example~\ref{ex: simple_rem}, where we assumed that $G=[0,1]^s$ and $K(x,A)=\pi(A)$ we
 obtain that $\a = 0$. In this case we get the well studied star-discrepancy for the uniform distribution on $[0,1]^s$, see for instance \cite{DiPi10}.
\end{remark}

 Theorem~\ref{thm: est_discr} gives an estimate of the star-discrepancy
 in terms of the discrepancy of the driver sequence and a quantity which depends on the
 transition kernel. We have seen in the previous corollary that for uniformly ergodic
 Markov chains we can further estimate the difference of the discrepancies  
 $D^\ast_{\mathscr{A}, \pi}(P_n)$ and $D^\ast_{\mathscr{A}, \varphi}(\mathcal{U}_n)$.
 Let us mention
 here that for geometrically ergodic transition kernel
 one can prove a similar bound.
 
In Corollary~\ref{cor_main_with_delta} we state a bound on
$D^\ast_{\mathscr{A}, \pi}(P_n)$ of order $\mathcal{O}(n^{-1/2} (\log n)^{1/2})$,
and by Corollary~\ref{coro: D_U_almost_D_P} we have that the pull-back discrepancy
of the driver sequence satisfies the same convergence order.

 From Corollary~\ref{coro: D_U_almost_D_P}
 and Theorem~\ref{thm_int_error}
 in Appendix~\ref{sec_int_error}
 we now obtain the following Koksma-Hlawka inequality
 (cf. \cite[Proposition~2.18]{DiPi10}) in terms of the pull-back discrepancy.

 \begin{corollary}[Koksma-Hlawka inequality for uniformly ergodic Markov chains]\label{cor_KH_inequality}
 Let us assume that the conditions of Theorem~\ref{thm: est_discr} are satisfied.
 Further let $\a\in[0,1)$ and $M\in(0,\infty)$ and
 assume that the transition kernel is uniformly
 ergodic with $(\a,M)$.
 With a measure $\rho$ on $\R^d$
 let $H_1$ denote the space
 of functions $f:G \to \mathbb{C}$ permitting
 the representation
 \[
  f(x) = f_0 + \int_{\R^d} \mathbf{1}_{(-\infty,z)_G}(x) \widetilde f(z) \, \rho(\rd z)
 \]
 for some $f_0 \in \mathbb{C}$ and $\widetilde f\in L_1(\R^d,\rho)$ with finite
 \[
  \norm{f}_{H_1} = \abs{f_0} + \int_{\R^d}
  \abs{\widetilde f(z)} \, \rho(\rd z).
 \]
 Then for all $f \in H_1$ we have
 \begin{align*}
\left|\int_G f(x) \pi(\rd z) - \frac{1}{n} \sum_{i=1}^n f(x_i) \right| \le \left( D^\ast_{\mathscr{A}, \varphi}(\mathcal{U}_n) + \frac{\a M }{n(1-\a)} \right) \|f\|_{H_1}.
\end{align*}
 \end{corollary}

Again, in the setting of Example~\ref{ex: simple_rem} for direct simulation we have $\alpha = 0$ and we obtain the Koksma-Hlawka inequality
\begin{equation*}
\left|\int_G f(x) \pi(\rd z) - \frac{1}{n} \sum_{i=1}^n f(x_i) \right| \le D^\ast_{\mathscr{A}, \varphi}(\mathcal{U}_n) \|f\|_{H_1}.
\end{equation*}

\section{On the existence of good driver sequences}\label{sec_exist}

In this section we show the existence of finite sequences $\mathcal{U}_n = (u_1, u_2,\ldots, u_n) \in [0,1]^{ns}$ such that
\begin{align*}
D^\ast_{\mathscr{A},\varphi}(\mathcal{U}_n)
\quad \text{and} \quad
D^\ast_{\mathscr{A}, \pi}(P_n)
\end{align*}
converge to $0$  if the transition kernel is uniformly ergodic
and $P_n$ is given by \eqref{eq: x_i_by_driver_seq}.
The main result is proven for $D^\ast_{\mathscr{A}, \pi}(P_n)$.
The result with respect to $D^\ast_{\mathscr{A}, \varphi}(\mathcal{U}_n)$ holds by Theorem~\ref{thm: est_discr}.\\

The concept of a $\delta$-cover will be useful (cf. \cite{Gn08} for a discussion of $\delta$-covers, bracketing numbers and Vapnik-\v{C}ervonenkis dimension).
\begin{definition}
Let $\mathscr{A} \subseteq \mathcal{B}(G)$ be a set of test sets.
A finite subset $\Gamma_\delta \subseteq \mathscr{A}$ is called a $\delta$-cover
of $\mathscr{A}$ with respect to $\pi$
 if for every $A \in \mathscr{A}$ there are sets $C, D \in \Gamma_\delta$ such that
\begin{equation*}
C \subseteq A \subseteq D
\end{equation*}
and
\begin{equation*}
\pi(D \setminus C) \le \delta.
\end{equation*}
\end{definition}

\begin{remark}
The concept of a $\delta$-cover is motivated by the following result.
Let us assume that $\Gamma_\delta$ is a $\delta$-cover of $\mathscr{A}$.
Then, for all $\{z_1,\dots,z_n\}$, the following discrepancy inequality holds
\begin{equation*}
\sup_{A \in \mathscr{A}} \left|\frac{1}{n} \sum_{i=1}^n 1_{z_i \in A} - \pi(A) \right|
\le \max_{C \in \Gamma_\delta} \left|\frac{1}{n} \sum_{i=1}^n 1_{z_i \in C} - \pi(C) \right| + \delta.
\end{equation*}

\begin{proof}
Let $A \in \mathscr{A}$ and $B \subseteq A \subseteq C$ be such that $\pi(C \setminus B) \le \delta$. Then
\begin{align*}
\frac{1}{n} \sum_{i=1}^n 1_{z_i \in A} - \pi(A) \le \frac{1}{n} \sum_{i=1}^n 1_{z_i \in C} - \pi(C) + \delta
\end{align*}
and
\begin{align*}
\frac{1}{n} \sum_{i=1}^n 1_{z_i \in A} - \pi(A) \ge \frac{1}{n} \sum_{i=1}^n 1_{z_i \in B} - \pi(B) - \delta.
\end{align*}
Thus the result follows.
\end{proof}
\end{remark}

Let us introduce the notation
$\Delta_{n,A,\varphi,x} =\Delta^{\loc}_{n,A,\varphi}(x; u_1,\dots,u_n)$
 and note that
\begin{align}
\Delta_{n,A,\varphi,x} \notag
& =\Delta^{\loc}_{n,A,\varphi}(x; u_1,\dots,u_n) \\ \notag
& =\frac{1}{n}  \sum_{i=1}^n \left[1_{(u_1,\ldots, u_i) \in B_i(x,A)} - \pi(A) \right] \\
\label{al: repr_delta}
& =\frac{1}{n}  \sum_{i=1}^n \left[1_{\varphi_i(x; u_1,\ldots, u_i) \in A} - \pi(A) \right].
\end{align}

\begin{lemma}  \label{lem: same_conc}
 Let $K$ be a transition kernel with stationary distribution $\pi$. Let $\varphi$ be an
 update function of $K$. Let $X_1, X_2, \dots, X_n$ be given by a Markov chain
 with transition kernel $K$ and $X_0=x$. Then for any $A\in\mathcal{B}(G)$ and $c>0$ we obtain
 \begin{equation}
  \P[ \abs{\Delta_{n,A,\varphi,x}} \geq c  ]
= \P_{x,K} \left[ \abs{\frac{1}{n} \sum_{i=1}^n 1_{X_i\in A}-\pi(A)} \geq c \right],
\end{equation}
  where $\P$ is the probability measure for the uniform distribution in $[0,1]^{ns}$ and $\P_{x,K}$ is
  the joint probability of $X_1,\dots,X_n$ with $X_0=x$.
\end{lemma}
 \begin{proof}
  Let
  \[
      J(A,c) = \left \{ (z_1,\dots,z_n)\in G^n \colon  \abs{\frac{1}{n} \sum_{i=1}^n 1_{z_i\in A}-\pi(A)} \geq c  \right\}.
  \]
  Set
  \[
   F(x_1,\dots,x_n)
= 1_{(x_1,\dots,x_n)\in J(A,c)}
= \begin{cases}
    1    & \abs{\frac{1}{n} \sum_{i=1}^n 1_{x_i\in A}-\pi(A)} \geq c,\\
    0    & \text{otherwise}.			
  \end{cases}
  \]
  By
  \[
   \EE_{x,K}(F(X_1,\dots,X_n)) = \P_{x,K} (J(A,c)),
  \]
   Lemma~\ref{lem: same_expect} and \eqref{al: repr_delta} the assertion is proven.
 \end{proof}

The next result follows from \cite{GlOr02} and
gives us a Hoeffding inequality for uniformly ergodic Markov chains.
For the convenience of the reader we provide a proof in Appendix~\ref{app: proof}.
\begin{proposition}[Hoeffding inequality for uniformly ergodic Markov chains]
\label{prop: Hoeffd}
Assume that the transition kernel $K$ is uniformly ergodic with $(\a,M)$.
 Let $X_1, X_2, \dots, X_n$ be given by a Markov chain
 with transition kernel $K$ and $X_0=x$.
 Then for any $A\in\mathcal{B}(G)$ and $c>0$ we obtain
 \begin{equation}
  \P_{x,K} \left[ \abs{\frac{1}{n} \sum_{i=1}^n 1_{X_i\in A}-\pi(A)} \geq c \right]
\leq 2
\exp\left( -\frac{(1-\a)^2}{M^2} \frac{( nc - \frac{2M}{1-\a} )^2}{ 8n}\right),
\end{equation}
where $n\geq  \frac{4M}{(1-\a)c}$.
\end{proposition}

\subsection{Monte Carlo rate of convergence}\label{subsec_main}

We now show that for every starting point $x_0$ and every $n$ there exists a finite sequence
$u_1, u_2, \ldots, u_n \in [0,1]^s$ such that the discrepancy of
the corresponding Markov chain converges approximately with order $n^{-1/2}$.
The main idea to prove the existence result is to use probabilistic arguments.
We apply a Hoeffding inequality for Markov chains to the local discrepancy function
for a fixed test set to show that the probability of point sets with small local discrepancy is large.
We then extend this result to the local discrepancy
for all sets in the $\delta$-cover and finally to all test sets.
Using Corollary~\ref{coro: D_U_almost_D_P} we are also able to
obtain a result for the pull-back discrepancy of the driver sequence.
In cases where there are $\delta$-covers with $|\Gamma_\delta| \le C \delta^{-\kappa}$ 
for some constants $C, \kappa> 0$ independent of $\delta$ 
(see Appendix~\ref{sec_delta_covers} for an example), 
the proof of the following theorem shows, in particular, 
that if the finite driver sequence is chosen at random from the uniform distribution, 
the discrepancy of the induced point set $P_n$ 
converges with high probability with almost the Monte Carlo rate.

\begin{theorem}\label{thm_main}
Let $K$ be a transition kernel with stationary distribution $\pi$
defined on a set $G \subseteq \mathbb{R}^d$. Assume that the transition kernel is uniformly
ergodic with $(\a,M)$. Let $\mathscr{A}\subseteq \mathcal{B}(G)$ be a set of test sets.
Assume that for every $\delta > 0$ there exists a set $\Gamma_\delta \subseteq \mathcal{B}(G)$
with $|\Gamma_\delta| < \infty$ such that $\Gamma_\delta$ is a $\delta$-cover of $\mathscr{A}$
with respect to $\pi$.
Let $\varphi$ be an update function for $K$.
Then, for any $x_0=x$ there exists a driver sequence $u_1, u_2, \ldots, u_n \in [0,1]^s$
such that $P_n=\{x_1,\dots,x_n\}$ given by
\begin{equation*}
 \jd{ x_i = \varphi(x_{i-1};u_i) } = \varphi_i(x_0; u_1,\dots,u_i), \quad i=1,\dots,n
\end{equation*}
satisfies
\begin{align*}
D^\ast_{\mathscr{A}, \pi}(P_n)
&  \le \frac{8M}{1-\a}\frac{\sqrt{\log |\Gamma_\delta|}}{\sqrt{n}} + \delta.
\end{align*}
\end{theorem}
\begin{proof}
Let $A\in\mathscr{A}$ and $x_0=x\in G$.
By Lemma~\ref{lem: same_conc} and Proposition~\ref{prop: Hoeffd}
we obtain for any
$c_n \geq \frac{4M}{n(1-\a)}$
that
\begin{equation}  \label{eq: hoeffd_for_us}
\mathbb{P}\left[\abs{  \Delta_{x,n,A,\varphi}}  \leq c_n \right]
\geq  1- 2 \exp\left( -\frac{(1-\a)^2}{M^2} \frac{( nc_n - \frac{2M}{1-\a} )^2}{ 8n}\right).
\end{equation}
Let
\begin{equation*}
   \widehat{\Gamma}_{\delta}=\{D\setminus C : C \subseteq A \subseteq D,\, C,D \in\Gamma_\delta \}.
\end{equation*}
Set $m= |\widehat{\Gamma}_\delta |$.
If we have for all $A\in \widehat{\Gamma}_{\delta}$ that
\begin{equation} \label{eq: for_all_sets}
 \mathbb{P}\left[\abs{  \Delta_{x,n,A,\varphi}}  \leq c_n \right] >  1 - \frac{1}{m},
\end{equation}
then there exists a finite sequence $u_1,\dots,u_n \in [0,1]^s$ such that
\begin{equation}  \label{eq: gamma_prime}
  \max_{A \in \widehat{\Gamma}_\delta} \abs{  \Delta_{x,n,A,\varphi}}  \leq c_n.
\end{equation}
By \eqref{eq: hoeffd_for_us} we obtain for
\[
 c_n=\frac{4M}{1-\a}\frac{\sqrt{2\log(2m)}}{\sqrt{n}}
\]
that \eqref{eq: for_all_sets} holds and we get the desired result for any $A\in \widehat{\Gamma}_\delta$.
Now we extend the result from $\widehat{\Gamma}_\delta$ to $\mathscr{A}$.
For $A\in\mathscr{A}$, there are $C,D \in \Gamma_\delta$ such that $C \subseteq A \subseteq D$
and $\pi(D\setminus C)\leq \delta$, since $\Gamma_\delta$ is a $\delta$-cover.
Hence we get
\begin{align*}
& \left|\frac{1}{n} \sum_{i=1}^n \left[ 1_{(u_1,\ldots, u_i) \in B_i(x,A) } - \pi(A) \right] \right|\\
=&  \left|\frac{1}{n} \sum_{i=1}^n \left[ 1_{(u_1,\ldots, u_i) \in B_i(x,D) } - \pi(D) \right]
-\frac{1}{n} \sum_{i=1}^n \left[ 1_{(u_1,\ldots, u_i) \in B_i(x,D\setminus A) } - \pi(D\setminus A) \right] \right|\\
\leq &  \left|\frac{1}{n} \sum_{i=1}^n \left[ 1_{(u_1,\ldots, u_i) \in B_i(x,D) } - \pi(D) \right] \right|
   +\left|\frac{1}{n} \sum_{i=1}^n  \left[ 1_{(u_1,\ldots, u_i) \in B_i(x,D\setminus A) } - \pi(D\setminus A) \right] \right| .
\end{align*}
Let
\[
 I_1:=\left|\frac{1}{n} \sum_{i=1}^n 
 \left[ 1_{(u_1,\ldots, u_i) \in B_i(x,D) } - \pi(D) \right] \right| 
\]
and
\[
 I_2:=\left|\frac{1}{n} \sum_{i=1}^n
\left[ 1_{(u_1,\ldots, u_i) \in B_i(x,D\setminus A) } - \pi(D\setminus A) \right] \right|.
\]
By $D\in \widehat{\Gamma}_{\delta}$ we have
\begin{equation*}
  I_1\leq\max_{A \in \widehat{\Gamma}_\delta} \left| \Delta_{n,A,\varphi,x} \right| \leq c_n.
\end{equation*}
Furthermore
\begin{align*}
I_2
  &\;\, = \left|\frac{1}{n} \sum_{i=1}^n 1_{(u_1,\ldots, u_i) \in B_i(x,D\setminus A)} - \pi(D\setminus C)  + \pi(D\setminus C) -\pi(D\setminus A) \right| \\
     &\;\, \leq \left|\frac{1}{n} \sum_{i=1}^n
\left[ 1_{(u_1,\ldots, u_i) \in B_i(x,D\setminus C)} - \pi(D\setminus C)\right]\right|
    + \left| \pi(D\setminus C) - \pi(D \setminus A) \right| \\
     &\;\, \leq c_n +  \delta.
\end{align*}
The last inequality follows by the $\delta$-cover property, \eqref{eq: gamma_prime}
and the fact that $D\setminus C \in \widehat{\Gamma}_\delta$.
Finally note that $m= |\widehat{\Gamma}_{\delta} | \leq |\Gamma_\delta |^2/2$ which completes the proof.
\end{proof}

Using Corollary~\ref{coro: D_U_almost_D_P} we can also state Theorem~\ref{thm_main} in terms of the driver sequence.

\begin{corollary}  \label{coro: push_back_push_forward}
Let $K$ be a transition kernel with stationary distribution $\pi$
defined on a set $G \subseteq \mathbb{R}^d$. Assume that the transition kernel is uniformly
ergodic with $(\a,M)$. Let $\mathscr{A}\subseteq \mathcal{B}(G)$ be a set of test sets.
Assume that for every $\delta > 0$ there exists a set $\Gamma_\delta \subseteq \mathcal{B}(G)$
with $|\Gamma_\delta| < \infty$ such that $\Gamma_\delta$ is a $\delta$-cover of $\mathscr{A}$ with respect to $\pi$.
Let $\varphi$ be an update function for $K$.
Then for any $x_0=x$ there exists a driver sequence $u_1, u_2, \ldots, u_n \in [0,1]^s$
such that
\begin{equation*}
D^\ast_{\mathscr{A}, \varphi}(\mathcal{U}_n) \le  \frac{8M}{1-\a}\frac{\sqrt{\log | \Gamma_\delta|}}{\sqrt{n}} + \delta + \frac{\alpha M}{n (1-\alpha)}.
\end{equation*}

Let $P_n=\{x_1,\dots,x_n\}$ given by
\begin{equation*}
 x_i =  \varphi(x_{i-1};u_i)  = \varphi_i(x_0;u_1,\dots,u_i), \quad i=1,\dots,n.
\end{equation*}
Then $P_n$ satisfies
\begin{align*}
D^\ast_{\mathscr{A}, \pi}(P_n)
&  \le \frac{8M}{1-\a}\frac{\sqrt{\log |\Gamma_\delta |}}{\sqrt{n}} + \delta + \frac{2 \alpha M}{n (1-\alpha)}.
\end{align*}
\end{corollary}

This corollary has two consequences. One is the existence of a driver sequence with small pull-back discrepancy. The second is that if one can construct such a sequence with small 
pull-back discrepancy,
then the Markov chain which one obtains using this driver sequence also has small discrepancy. 
Thus the pull-back discrepancy is a sufficient criterion for the construction of good driver sequences.

Theorem~\ref{thm_main} and Corollary~\ref{coro: push_back_push_forward}
depend on $\delta$ and the size of the $\delta$-cover $\Gamma_\delta$.
For a certain set of test sets we have the following result.
\begin{corollary}\label{cor_main_with_delta}
Let $K$ be a transition kernel with stationary distribution $\pi$
defined on a set $G \subseteq \mathbb{R}^d$. Assume that the transition kernel is uniformly
ergodic with $(\a,M)$.
 Let the set of test sets $\mathscr{A}\subseteq \mathcal{B}(G)$ be given by
 \[
 \mathscr{A} = \{ (-\infty,x)\cap G \mid x\in \bar{\R}^d\},
\]
where $\bar{\R}^d=(\R\cup\{ \infty,-\infty\})^d$.
Let $\varphi$ be an update function for $K$.
Then for any $x_0=x$ there exists a driver sequence $u_1, u_2, \ldots, u_n \in [0,1]^s$
and an absolute constant $c>0$
such that
\begin{equation*}
D^\ast_{\mathscr{A}, \varphi}(\mathcal{U}_n)
\le
\frac{8M}{1-\a}\frac{\sqrt{d\log(3+4c^2n)}}{\sqrt{n}}
+ \frac{\sqrt{d}}{\sqrt{n}} + \frac{\alpha M}{n (1-\alpha)}.
\end{equation*}
Let $P_n=\{x_1,\dots,x_n\}$ given by
\begin{equation*}
 x_i = \varphi(x_{i-1};u_i)  = \varphi_i(x_0; u_1,\dots,u_i), \quad i=1,\dots,n.
\end{equation*}
Then $P_n$ satisfies
\begin{align*}
D^\ast_{\mathscr{A}, \pi}(P_n) & \le
\frac{8M}{1-\a}\frac{\sqrt{d\log(3+4c^2n)}}{\sqrt{n}}
+ \frac{\sqrt{d}}{\sqrt{n}} + \frac{2\alpha M}{n (1-\alpha)}.
\end{align*}
\end{corollary}
\begin{proof}
The result follows by Lemma~\ref{lem_delta_cover_ex} in Appendix~\ref{sec_delta_covers},
which shows the existence of $\delta$-covers with $$|\Gamma_\delta| \le (4 + 3 c^2 d \delta^{-2})^d,$$
and Corollary~\ref{coro: push_back_push_forward}.
\end{proof}

\subsection{Optimality of the Monte Carlo rate}\label{subsec_optimal}

We now show that the exponent $-1/2$ of $n$ in Theorem~\ref{thm_main} cannot be improved in general.
We do so by specializing Theorem~\ref{thm_main} to the sphere $\mathbb{S}^d$.
Recall that $\mathbb{S}^d = \{x \in \mathbb{R}^{d+1}: \|x\|_{2} = \langle x, x \rangle^{1/2} = 1\}$,
where $\langle x, y \rangle$ denotes the standard inner product in $\mathbb{R}^{d+1}$.
A spherical cap $C(x,t) \subseteq \mathbb{S}^d$ with center $x \in \mathbb{S}^d$ and $-1 \le t \le 1$ is given by \begin{equation*}
C(x,t) = \{y \in \mathbb{S}^d: \langle x, y \rangle > t\}.                                                                                                                                                                                                                                                                                                                                                                                                               \end{equation*}
Let $\mathcal{C} = \{C(x,t): x \in \mathbb{S}^d, -1 \le t \le 1\}$ be the test set of spherical caps of $\mathbb{S}^d$. The normalized area of a spherical cap $C(x,t)$ for
$0 \le t \le 1$ is given by
\begin{equation*}
\pi(C(x,t)) = \frac{1}{2} \frac{B(1-t^2; d/2,1/2)}{B(1; d/2,1/2)},
\end{equation*}
where $B$ is the incomplete beta function
\begin{equation*}
B(1-t^2; d/2, 1/2) = \int_0^{1-t^2} z^{d/2-1} (1-z)^{-1/2} \rd z.
\end{equation*}
Then the spherical cap discrepancy of a point set $P_n = \{x_1, x_2,\ldots, x_n\} \subseteq \mathbb{S}^d$ is given by
\begin{equation*}
D^\ast_{\mathbb{S}^d, \mathcal{C}}(P_n) = \sup_{C \in \mathcal{C}} \left|\frac{1}{n} \sum_{i=1}^n 1_{x_i \in C} - \pi(C) \right|.
\end{equation*}

The following result is an application of Theorem~\ref{thm_main}.
For a proof of the corollary we refer to Appendix~\ref{ss: coro_sphere}.

\begin{corollary}\label{cor_sphere}
There exists an absolute constant $c > 0$ independent of $n$ and $d$ such that for each $n$ and $d$ there exists a set of points $P_n = \{x_1, x_2,\ldots, x_n\} \subseteq \mathbb{S}^d$ such that the spherical cap discrepancy satisfies
\begin{equation*}
D^\ast_{\mathbb{S}^d, \mathcal{C}}(P_n) \le c\, \frac{\sqrt{d} + \sqrt{(d+1) \log n}}{\sqrt{n}}.
\end{equation*}
\end{corollary}

We have shown the existence of points on $\mathbb{S}^d$
for which the spherical cap discrepancy is of order $\sqrt{\frac{\log n}{n}}$.
Since this result follows by specializing Theorem~\ref{thm_main} to the sphere,
any improvement of the exponent $-1/2$ of $n$ in Theorem~\ref{thm_main} would
yield an improvement of the exponent of $n$ in Corollary~\ref{cor_sphere}.
 However, it is known that the spherical cap discrepancy of any point set is at least $n^{-1/2 - 1/(2d)}$,
see \cite{Beck}. Thus, an exponent smaller than $-1/2$ in Corollary~\ref{cor_sphere}
would yield a contradiction to the lower bound on the spherical cap discrepancy for large enough $d$.
Thus, at this level of generality, the exponent of $n$ in Theorem~\ref{thm_main} cannot be improved.

We point out that a bound on the spherical cap discrepancy
can also be deduced from \cite[Theorem~4]{HeNoWaWo01}
by using a bound on the Vapnik-\v{C}ervonenkis dimension for $\mathcal{C}$.

\begin{appendix}

\section{Proof of Proposition~2}
\label{app: proof}
Since the statement and the assumptions in Proposition~\ref{prop: Hoeffd} 
are slightly different from those in \cite[Theorem~2]{GlOr02}, 
we prove the desired Hoeffding inequality by following
the arguments in \cite{GlOr02}.
We set $f(x)=1_{x\in A} - \pi(A)$ and obtain by the uniform ergodicity
\begin{align*}
 P^n f(x) = K^n(x,A) -\pi(A) \leq \min\{1,\alpha^n M\}.
\end{align*}
Let $g(x) = \sum_{n=0}^\infty P^n f(x)$ and note that $\norm{g}_\infty \leq M/(1-\alpha)$.
Here let us mention that even
\[
 \norm{g}_\infty 
 \leq 1+(1-\alpha)^{-1}+\frac{\log M}{\log \alpha^{-1}},
\]
which improves upon the dependence of $M$.
We have $g(x)-Pg(x) = f(x)$ for all $x\in G$. 
Let $D(X_{j-1},X_j) = g(X_j)-Pg(X_{j-1})$ for $j=2,\dots,n+1$,
then
\begin{align*}
& \quad \sum_{j=1}^n 1_{X_j\in A} - n\cdot \pi(A)
   = \sum_{j=1}^n f(X_j)
   = \sum_{j=1}^n (g(X_j)-Pg(X_j))\\
&  = g(X_1)-g(X_{n+1})+\sum_{j=2}^{n+1} D(X_{j-1},X_j)
   \leq 2\norm{g}_\infty + \sum_{j=2}^{n+1} D(X_{j-1},X_j).
\end{align*}
Thus, for $\theta > 0$ we obtain
\begin{align*}
 & \;\mathbb{E}_{x,K} \left[ \exp\left(\theta \sum_{j=1}^n 1_{X_j\in A} - \theta n\cdot \pi(A)\right)\right]\\
 & \leq \exp(2\theta\norm{g}_
 \infty)\;\mathbb{E}_{x,K} \left[\exp\left(\theta\sum_{j=2}^{n+1} D(X_{j-1},X_j)\right)\right]
\end{align*}
and
\begin{align*}
&\;\mathbb{E}_{x,K} \left[\exp\left(\theta\sum_{j=2}^{n+1} D(X_{j-1},X_j)\right)\right]\\
 =& \;\mathbb{E}_{x,K} \left[\exp\left(\theta\sum_{j=2}^{n} D(X_{j-1},X_j)\right) \cdot 
 \mathbb{E}_{x,K}(\exp(\theta D(X_n,X_{n+1}))\mid X_1,\dots,X_{n}) \right].
\end{align*}
 By virtue of \cite[Lemma~8.1]{Devetal96}, see also \cite[Equation (5)]{GlOr02}, we obtain
 \begin{align} \label{al: iter}
   & \;\mathbb{E}_{x,K}(\exp(\theta D(X_n,X_{n+1}))\mid X_1,\dots,X_{n}) \leq \exp(\theta^2\norm{g}_\infty^2 /2).
\end{align}
This can be repeated iteratively so that
\[
\mathbb{E}_{x,K} \left[ \exp\left(\theta \sum_{j=1}^n 1_{X_j\in A} - \theta n\cdot \pi(A)\right)\right]
\leq \exp(2\theta\norm{g}_\infty+n\theta^2\norm{g}_\infty^2 /2).
\]
Markov's inequality leads to
\[
 \mathbb{P}_{x,K} \left[ \frac{1}{n} \sum_{j=1}^n 1_{X_j\in A} - \pi(A) \geq c \right]
 \leq \exp(-\theta nc+2\theta\norm{g}_\infty+ n \theta^2\norm{g}_\infty^2 /2).
\]
The bound is best possible
for 
\[
\theta = \frac{nc-2\norm{g}_\infty}{n\norm{g}_\infty^2}. 
\]
Finally by $\norm{g}_\infty \leq M/(1-\alpha)$
and by repeating this analysis 
with $f=\pi(A)-1_{x\in A}$
we obtain the assertion.

\section{Delta-covers}\label{sec_delta_covers}

We need a deep result of the theory of empirical processes which follows from Talagrand~\cite[Theorem~6.6]{Ta94} and Haussler~\cite[Corollary~1]{Ha95}. For a more general version see also \cite[Theorem~4]{HeNoWaWo01}.

\begin{proposition}  \label{prop: talagrand_haussler}
 There exists an absolute constant $c>0$ such that for each cumulative distribution function $L$ on $(G,\mathcal{B}(G))$
 the following holds:
 For all $r\in \NN$ there exist $y_1,\dots,y_r \in G$ with
 \[
  \sup_{x\in \mathbb{Q}^d} \abs{ L(x) - \frac{1}{r} \sum_{i=1}^r 1_{y_i\in (-\infty,x)_G}  } \leq c\, \sqrt{d}\, r^{-1/2}.
 \]
\end{proposition}

In this subsection we study $\delta$-covers in $G$
with respect to the probability measure $\pi$.

\begin{lemma}\label{lem_delta_cover_ex}
Let $G \subseteq \mathbb{R}^d$
and let $(G, \mathcal{B}(G), \pi)$ be a probability space where $\mathcal{B}(G)$
is the Borel $\sigma$-algebra of $G$.
Assume that $\pi$ is absolutely continuous with respect to
the Lebesgue measure.
Define the set $\mathscr{A} \subseteq \mathcal{B}(G)$
of test sets by
$$\mathscr{A} = \{(-\infty, x)_G: x \in \bar{\mathbb{R}}^d\}.$$
Then for any $\delta > 0$ there exists a $\delta$-cover $\Gamma_\delta$ of $\mathscr{A}$
with $$|\Gamma_\delta| \le (3 + 4 c^2 d \delta^{-2})^d,$$
where $c>0$ is an absolute constant.
\end{lemma}

\begin{proof}
Let $\delta > 0$ be given and let $r \in \mathbb{N}$ be the smallest integer such that $2 c \sqrt{d} r^{-1/2} \le \delta$. By Proposition~\ref{prop: talagrand_haussler} there are points $y_1, \ldots, y_r \in G$ such that
\begin{equation}\label{discrepancy_pts}
\sup_{x \in \bar{\mathbb{R}}^d} \left|\pi((-\infty, x)_G) - \frac{1}{r} \sum_{i=1}^r 1_{y_i \in (-\infty, x)_G} \right|
\le c \, \sqrt{d} \, r^{-1/2} \le \frac{\delta}{2}.
\end{equation}
Let $y_i = (\eta_{i,1}, \ldots, \eta_{i,d})$. We now define the set
\begin{equation*}
\Gamma_\delta = \left\{\prod_{j=1}^d (-\infty, z_j)\cap G: z_j \in \{-\infty, \infty, \eta_{1,j}, \eta_{2,j} \ldots, \eta_{r,j} \} \mbox{ for } 1 \le j \le d \right\}.
\end{equation*}
The cardinality of $\Gamma_\delta$ satisfies
\begin{equation*}
|\Gamma_\delta| = (2 + r)^{d} \le (3 + 4 c^2 d \delta^{-2})^d.
\end{equation*}
It remains to show that $\Gamma_\delta$ is a $\delta$-cover of $\mathscr{A}$.

Let $z \in \bar{\R}^d$ be arbitrary.
Then there exist $(-\infty, x)_G, (-\infty, y)_G \in \Gamma_\delta$ such that
\[
 (-\infty, x]_G \subseteq (-\infty, z)_G \subseteq (-\infty, y)_G
\]
and
\begin{align*}
 (-\infty, x]_G \cap \{y_1,\ldots, y_r\}
& = (-\infty, z)_G \cap \{y_1,\ldots, y_r\}\\
& = (-\infty, y)_G \cap \{y_1,\ldots, y_r\}.
\end{align*}
Using \eqref{discrepancy_pts} we obtain
\begin{align*}
& \quad\; \pi((-\infty, y)_G \setminus (-\infty, x)_G)  \\
& \le  \left|\pi((-\infty, y)_G) - \frac{1}{r} \sum_{i=1}^r 1_{y_i \in (-\infty, y)_G}\right|
+ \left|\pi((-\infty, x]_G) - \frac{1}{r} \sum_{i=1}^r 1_{y_i \in (-\infty, x]_G}\right| \\
& \le  \delta.
\end{align*}
Thus $\Gamma_\delta$ is a $\delta$-cover.
\end{proof}

\section{Integration error}\label{sec_int_error}

In Appendix~\ref{sec_delta_covers} we considered test sets which are intersections of boxes with the state space $G$. We define a reproducing kernel $Q$ by
\begin{equation*}
Q(x,y) = 1 + \int_{\mathbb{R}^d} 1_{(-\infty, z)_G}(x) 1_{(-\infty, z)_G}(y) \rho(\rd z),
\end{equation*}
where $\rho$ is a  measure on ${\mathbb{R}^d}$
with $\int_{\mathbb{R}^d} \rho(\rd z) < \infty$.
The function $Q$ is symmetric $Q(x,y) = Q(y,x)$
and positive semi-definite, that is,
for any $x_1,\ldots, x_n \in G$ and complex
numbers $b_1,\ldots, b_n \in \mathbb{C}$ we have
\begin{equation*}
\sum_{k, \ell = 1}^n b_k \overline{b_\ell} Q(x_k, x_\ell)
= \left|\sum_{k=1}^n b_k\right|^2 + \int_{\mathbb{R}^d}
\left|\sum_{k=1}^n b_k 1_{(-\infty, z)_G}(x_k) \right|^2 \rho(\rd z) \ge 0,
\end{equation*}
where $\overline{b_\ell}$ denotes the complex conjugate of $b_\ell$.
Thus $Q$ uniquely defines a reproducing kernel Hilbert space $H_2 = H_2(Q)$ of functions defined on $G$.
See \cite{Ar50} for more information on reproducing kernels and reproducing kernel Hilbert spaces. In fact, the functions $f$ in $H_2$ permit the representation
\begin{equation}\label{eq_f_rep}
f(x) = f_0 + \int_{\mathbb{R}^d} 1_{(-\infty, z)_G}(x) \widetilde{f}(z) \rho(\rd z),
\end{equation}
for some $f_0 \in \mathbb{C}$ and function
$\widetilde{f} \in L_2({\mathbb{R}^d}, \rho)$,
which can for instance be shown using the same arguments
as in \cite[Appendix~A]{BrDi13}. The inner product in $H_2$ is given by
\begin{equation*}
\langle f, g \rangle = f_0 \overline{g_0}
+ \int_{\mathbb{R}^d} \widetilde{f}(z) \overline{\widetilde{g}(z)} \rho(\rd z).
\end{equation*}
With these definitions we have the reproducing property
\begin{equation*}
\langle f, Q(\cdot, y)\rangle = f_0
+ \int_{\mathbb{R}^d} \widetilde{f}(z) 1_{(-\infty, z)_G}(y) \rho(\rd z) = f(y).
\end{equation*}

For $1 \le q \le \infty$ we also
define the space $H_q$ of functions of
the form \eqref{eq_f_rep} for which
$\widetilde{f} \in L_q({\mathbb{R}^d}, \rho)$, with norm
\begin{equation*}
\|f\|_{H_q} = \left(|f_0|^q
+ \int_{\mathbb{R}^d} |\widetilde{f}(z)|^q \rho(\rd z) \right)^{1/q}.
\end{equation*}

We provide a simple example.
\begin{example}
Let $G = [0,1]$ and let $\rho$ be the Lebesgue measure, then
\begin{equation*}
Q(x,y) = 1 + \int_0^1 1_{[0,z)}(x) 1_{[0,z)}(y) \rd z = 1 + \min\{1-x, 1-y\}.
\end{equation*}
The function $\widetilde{f} = -f'$,
where $f'$ is the usual derivative of $f$,
and \eqref{eq_f_rep} is then
$f(x) = f_0 + \int_{0}^1 1_{[0,z)}(x) \widetilde{f}(z) \rd z = f_0 -\int_x^1 f'(z) \rd z = f_0 + f(x) - f(1)$.
Thus $f(1) = f_0$ and $H_q$ is the space of all absolutely continuous functions $f$ for which $f' \in L_q([0,1],\rho)$.
\end{example}

We have the following result concerning the integration error in $H_q$.

\begin{theorem}\label{thm_int_error}
Let $G \subseteq \mathbb{R}^d$ and $\pi$ be a probability measure on $G$.
Further let $\mathscr{A}  = \{(-\infty, x)_G: x \in G\}$.
We assume that $1 \le p, q \le \infty$ with $1/p + 1/q = 1$.
Then for $P_n=\{x_1, x_2, \ldots, x_n \} \subseteq G$ and for all $f \in H_q$ we have
\begin{equation*}
\left|\int_G f(z) \pi(\rd z) - \frac{1}{n} \sum_{i=1}^n f(x_i)\right| \le \|f\|_{H_q} D^\ast_{p, \mathscr{A}, \pi}(P_n),
\end{equation*}
where
\begin{equation*}
D^\ast_{p, \mathscr{A}, \pi}(P_n)
= \left( \int_{\mathbb{R}^d} \left| \int_G  1_{(-\infty, z)_G}(y) \pi(\rd y)  - \frac{1}{n} \sum_{i=1}^n 1_{(-\infty, z)_G}(x_i) \right|^p \rho(\rd z) \right)^{1/p},
\end{equation*}
and for $p=\infty$ let
\begin{equation*}
D^\ast_{\mathscr{A}, \pi}(P_n)
:= D^\ast_{\infty, \mathscr{A}, \pi}(P_n) = \sup_{z \in G}
\left| \int_G  1_{(-\infty, z)_G}(y) \pi(\rd y)  - \frac{1}{n} \sum_{i=1}^n 1_{(-\infty, z)_G}(x_i) \right|.
\end{equation*}
\end{theorem}

\begin{proof}
Let
\begin{equation*}
e(f,P_n) = \int_G f(z) \pi(\rd z) - \frac{1}{n} \sum_{i=1}^n f(x_i)
\end{equation*}
denote the quadrature error when approximating the integral $\int_G f(z) \pi(\rd z)$ by $\frac{1}{n} \sum_{i=1}^n f(x_i)$ where $P_n = \{x_1, x_2, \ldots, x_n\}$.

Let $h(x) = \int_{G} Q(x, y) \pi(\rd y) - \frac{1}{n} \sum_{i=1}^n Q(x, x_i)$, then we have
\begin{align*}
h(x) = \int_{\mathbb{R}^d} 1_{(-\infty, z)_G}(x) \left( \int_G  1_{(-\infty, z)_G}(y) \pi(\rd y)  - \frac{1}{n} \sum_{i=1}^n 1_{(-\infty, z)_G}(x_i) \right) \rho(\rd z)
\end{align*}
and therefore $h \in H_p$ for any $1 \le p \le \infty$. Let \begin{equation*}
\widetilde{h}(z) = \int_G  1_{(-\infty, z)_G}(y) \pi(\rd y)  - \frac{1}{n} \sum_{i=1}^n 1_{(-\infty, z)_G}(x_i).
\end{equation*}
Further, for $f \in H_q$ we have $\widetilde{f} \in L_q({\mathbb{R}^d}, \rho)$ and thus
\begin{align*}
e(f,P_n) = \int_{\mathbb{R}^d} \widetilde{f}(z) \widetilde{h}(z) \rho(\rd z).
\end{align*}

Using H\"older's inequality we have
\begin{align*}
|e(f,P_n)| \le & \int_{\mathbb{R}^d}
\left| \widetilde{f}(z) \right| \left|\widetilde{h}(z)\right| \rho(\rd z) \\
\le & \left(\int_{\mathbb{R}^d} \left|\widetilde{f}(z)\right|^q \rho(\rd z)
\right)^{1/q}  \left(\int_{\mathbb{R}^d} \left|\widetilde{h}(z)\right|^p \rho(\rd z) \right)^{1/p},
\end{align*}
where $1 \le p, q \le \infty$ are H\"older conjugates $1/p + 1/q = 1$, with the obvious modifications for $p, q = \infty$. Thus the result follows.
\end{proof}

Thus we can use the bounds from the theorems above to obtain a bound on the integration error $|e(f,P_n)|$, where $P_n$ is the set of points from the Markov chain, for functions $f$ with representation \eqref{eq_f_rep} and $\|f\|_{H_1} < \infty$.

\section{Delta-covers for the sphere}\label{ss: coro_sphere}

We use Theorem~\ref{thm_main} where $\pi$ is the normalized Lebesgue surface measure on the sphere $\mathbb{S}^d$. Let $\psi: [0,1]^d \to \mathbb{S}^d$ be an area-preserving mapping
from $[0,1]^d$ to $\mathbb{S}^d$ (i.e., a generator function), see \cite{Fang_Wang}, and let the update function
$\varphi: \mathbb{S}^d \times [0,1]^d \to \mathbb{S}^d$ be given by
$$\varphi(x,u) = \psi(u).$$
The transition kernel is given by $K(x,A) = \pi(A)$ which is uniformly ergodic with $(\a,M)$ for $\a=0$ and $M=1$.\\

In order to obtain a bound on the spherical cap discrepancy using Theorem~\ref{thm_main},
it remains to construct a $\delta$-cover on $\mathbb{S}^d$ of suitable size.
We construct a $\delta$-cover $\Gamma_\delta$ by specifying a set of centers and heights in the following.

\begin{lemma}\label{lem_delta_cover_sphere}
Let $\mathbb{S}^d \subseteq \mathbb{R}^{d+1}$ denote the $d$-dimensional sphere.
Let $\mathscr{C} = \{C(x,t): x \in \mathbb{S}^d, -1 \le t \le 1\}$
denote the set of spherical caps of $\mathbb{S}^d$. Then for any $\delta > 0$
there exists a $\delta$-cover $\Gamma_\delta$ of $\mathscr{C}$ with respect to the normalized surface Lebesgue measure on $\mathbb{S}^d$ with $|\Gamma_\delta| \le c d^{d+1} \delta^{-2(d+1)}$, where $c > 0$ is a constant independent of $d$ and $\delta$.
\end{lemma}

The result of Corollary~\ref{cor_sphere} follows now from Theorem~\ref{thm_main} and Lemma~\ref{lem_delta_cover_sphere} by setting $\delta = d^{1/2} n^{-1/2}$. The remainder of this subsection is concerned with the proof of Lemma~\ref{lem_delta_cover_sphere}.

Let $y_1, y_2, \ldots, y_N \in \mathbb{S}^d$ be given such that
\begin{equation}\label{mesh_norm}
\sup_{x \in \mathbb{S}^d} \min_{1 \le i \le N} \|x-y_i\|_2 \le c_d N^{-1/d},
\end{equation}
where $c_d > 0$ is a constant depending only on $d$. The existence of such point sets follows, for instance, from \cite{L07}.
Therein an equal area partition of $\mathbb{S}^d$ into $N$ parts was shown with diameter bounded by $c_d N^{-1/d}$.
Thus by taking one point in each partition we obtain \eqref{mesh_norm}.
Indeed, from the proof of \cite[Theorem~2.6]{L07} we obtain that the constant $c_d$ can be chosen as
\begin{equation*}
c_d = 8 \left(\frac{d \sqrt{\pi}\; \Gamma(d/2)}{\Gamma((d+1)/2)}\right)^{1/d} \le 8 d^{1/d} \pi^{1/(2d)} \le 8 \cdot 3^{1/3} \sqrt{\pi} < 21,
\end{equation*}
where $\Gamma(x)=\int_0^\infty  t^{x-1} e^{-t} \rd t$ denotes the Gamma function.
For $x, y \in \mathbb{S}^d$ we have $\|x-y\|_2^2 = 2(1-\langle x, y \rangle)$.

Let $v = \langle x, y \rangle$. Then we obtain the following result.
\begin{lemma}
 We have $C(x,t) \subseteq C(y,u)$ if and only if $v = \langle x, y \rangle > u$ and
\begin{equation}\label{eq_subsets}
t^2 + u^2 + v^2 -2tuv > 1.
\end{equation}
\end{lemma}
\begin{proof}
 The condition $v > u$ ensures that $x \in C(y,u)$. Let $z \in C(x,t)$, that is, $\langle z, x \rangle > t$. Then $z \in C(y,u)$ if and only if $\langle z, y \rangle > u$. The point $z$ is furthest from $y$ (as measured by the Euclidean distance) if it lies on the great circle containing $x$ and $y$. Assuming that $x,y,z$ all lie on the same great circle such that $x$ is between $y$ and $z$, we have
\begin{align*}
\|y-z\| = & 2\sin \left(\arcsin \frac{\|x-y\|}{2} + \arcsin \frac{\|x-z\|}{2}\right) \\ = & \|x-y\| \sqrt{1 - \|x-z\|^2/4} + \|x-z\| \sqrt{1-\|x-y\|^2/4}.
\end{align*}
The result now follows by using $\|x-y\|^2 = 2(1-v)$ and $ \|x-z\|^2 < 2(1-t)$.
\end{proof}
The next lemma gives us a $\delta$-cover of $\mathcal{C}$ with respect to $\pi$.

\begin{lemma} \label{lem: delta_cover_sphere}
 Let
\[
N = \left \lceil \frac{35^d c_d^d}{B^{2d}(1;d/2,1/2) \delta^{2d}} \right\rceil.
\]
 Let $M = \lfloor N^{1/d}/c_d \rfloor$ and $T = \{-1+ k/M: k = 0,1,\ldots, 2M\}$.
 Then the set
$$\Gamma_\delta = \{C(y_i,t): 1 \le i \le N, t \in T\}$$
is a $\delta$-cover of $\mathcal{C}$ with respect to $\pi$.
  \end{lemma}

\begin{proof}
 Let $C(x,t) \in \mathcal{C}$ be an arbitrary spherical cap.
Let $y_i$ be such that $\|x-y_i\|_2 \le C_d N^{-1/d} \le 1/M$, thus
\begin{equation*}
\langle x, y_i \rangle > 1-\frac{1}{2M^2}.
\end{equation*}
Let $u, w \in T$ be such that $u + 2/M \le t \le w - 2/M$ and $w-u \le 5/M$.

We now show that for this choice we have $C(x,t) \subseteq C(y_i, u)$. First assume that $u \ge 0$. Then using \eqref{eq_subsets} with $v > 1- 1/(2M^2)$ and $t-u \ge 2/M$ we obtain
\begin{align*}
t^2+u^2 + v^2-2 tuv \ge & t^2-2tu +u^2 + v^2 \\ \ge & (t-u)^2 + (1-1/(2M^2))^2 \\ \ge & 4/M^2 + 1 - 1/M^2 + 1/(4M^4) > 1.
\end{align*}
Now assume that $u < 0$. Then
\begin{align*}
t^2+u^2 + v^2-2 tuv \ge & t^2-2tu +u^2 + v^2 - tu/M^2 \\ \ge & (t-u)^2 + (1-1/(2M^2))^2 - 1/M^2 \\ \ge & 4/M^2 + 1 - 1/M^2 + 1/(4M^4) - 1/M^2 > 1.
\end{align*}
Thus $C(x,t) \subseteq C(y_i,u)$. We now show that $C(y_i, w) \subseteq C(x,t)$. If $w=1$ we have $C(y_i,w) = \emptyset$ in which case the result holds trivially. Thus we can assume that $t < 1-2/M$, which implies that $y_i \in C(x,t)$. We use \eqref{eq_subsets} again with $v > 1-1/(2M^2)$ and $|u-t| \ge 1/M$. Thus the result follows by the same arguments
as in the previous case.

Thus we have $C(y_i, w) \subseteq C(x,t) \subseteq C(y_i, u)$ with $w-u \le 5/M$. For $w,u \ge 0$ we have
\begin{align*}
& 2 B(1;d/2,1/2)\, \pi(C(y_i,u) \setminus C(y_i,w)) \\
= & \int_{1-w^2}^{1-u^2} z^{d/2-1} (1-z)^{-1/2} \rd z \\
\le & \int_{\max\{0, 1- u^2 - 10u/M - 25/M^2\}}^{1-u^2} z^{d/2-1} (1-z)^{-1/2} \rd z \\ \le & \sup_{10/M+25/M^2 \le r \le 1} \int_{r-10/M-25/M^2}^r z^{d/2-1} (1-z)^{-1/2} \rd z \\ \le & \int_{1-10/M-25/M^2}^1 (1-z)^{-1/2} \rd z \\ \le & \sqrt{10/M+25/M^2} \le \sqrt{35/M}.
\end{align*}
Thus, in general we have
\begin{equation*}
\pi(C(y_i,u) \setminus C(y_i,w))
\le \frac{\sqrt{35/M}}{B(1; d/2, 1/2)}
\le \frac{\sqrt{35}}{B(1; d/2, 1/2)} \frac{\sqrt{c_d}}{N^{1/(2d)}}.
\end{equation*}
The last expression is bounded by $\delta$
for $$N = \left \lceil \frac{35^d c_d^d}{B^{2d}(1;d/2,1/2) \delta^{2d}} \right\rceil.$$
\end{proof}
\begin{lemma}  \label{lem: card_delta_cover_sphere}
 Assume that $N^{1/d}/c_d > 1/2$ (otherwise \eqref{mesh_norm} is trivial).
 Then we have for $\delta = d^{1/2} n^{-1/2}$ that there exist an absolute constant $c>0$, such that
 \[
   |\Gamma_\delta| \le c\, n^{d+1}.
 \]
\end{lemma}
\begin{proof}
We have
\begin{align*}
|\Gamma_\delta| \le & N (2M+1) \\ \le & N (2N^{1/d} /c_d + 1) \\  \le & 4 N^{1+1/d} / c_d \\
\le & \frac{8 \cdot 35^{d+1} c_d^{d}}{B^{2(d+1)}(1;d/2,1/2) \delta^{2(d+1)}}.
\end{align*}
Thus for $\delta = d^{1/2} n^{-1/2}$ we obtain
\begin{equation*}
|\Gamma_\delta| \le \frac{n^{d+1}}{d^{d+1}} \frac{8 \cdot 35^{d+1} c_d^{d}}{B^{2(d+1)}(1;d/2,1/2)}
\le \frac{n^{d+1}}{d^{d+1}} \frac{280\cdot 735^d (\Gamma((d+1)/2))^{2(d+1)}}{(\Gamma(d/2) \Gamma(1/2))^{2(d+1)}},
\end{equation*}
where $\Gamma$ is the Gamma function. Using Stirling's formula for the Gamma function
\begin{equation*}
\Gamma(z) = \sqrt{\frac{2\pi}{z}}\left(\frac{z}{\mathrm{e}}\right)^z (1 + \mathcal{O}(z^{-1}))
\end{equation*}
we obtain that there is an absolute constant $c > 0$ such that
$
 |\Gamma_\delta| \le c\, n^{d+1}.
$
\end{proof}

\end{appendix}

\section*{Acknowledgements}

Josef Dick is the recipient of an Australian Research Council Queen
Elizabeth II Fellowship (project number DP1097023). 
Daniel Rudolf was supported partially by the ARC 
Discovery Grant DP120101816, by the DFG priority program 1324 
and the DFG Research training group 1523. 
Houying Zhu was supported partially by the ARC Discovery Grant DP120101816.

\end{document}